\documentclass{article}
\usepackage{hyperref}
\usepackage{algorithm}
\usepackage{algorithmicx}
\usepackage{algpseudocode}
\usepackage[finalreport]{neurips_2025}
\usepackage{amsthm}

\newtheorem{claim}{Claim}   
\newtheorem{theorem}{Theorem}[section]        
\newtheorem{lemma}[theorem]{Lemma}            
\newtheorem{construction}{Construction}

\theoremstyle{definition}
\newtheorem{definition}[theorem]{Definition}

\theoremstyle{definition}
\newtheorem*{remark}{Remark}
\usepackage[utf8]{inputenc} 
\usepackage[T1]{fontenc}    
\usepackage{hyperref}       
\usepackage{url}            
\usepackage{booktabs}       
\usepackage{amsfonts}       
\usepackage{amsmath}        
\usepackage{nicefrac}       
\usepackage{microtype}      
\usepackage{xcolor}         
\usepackage{graphicx}       
\usepackage{enumitem}       

\title{Finding 4-Additive Spanners:  Faster, Stronger, and Simpler}
\author{%
  Chuhan Qi\thanks{Undergraduate student.} \\
  Institute for Interdisciplinary Information Sciences\\
  Tsinghua University\\
  Beijing 100084, China \\
  \texttt{qch22@mails.tsinghua.edu.cn}
}

\begin{document}
\raggedbottom
\maketitle

\begin{abstract}
Additive spanners are fundamental graph structures with wide applications in network design, graph sparsification, and distance approximation. In particular, a $4$-additive spanner is a subgraph that preserves all pairwise distances up to an additive error of $4$. In this paper, we present a new deterministic algorithm for constructing $4$-additive spanners that matches the best known edge bound of $\tilde{O}(n^{7/5})$  (up to polylogarithmic factors), while improving the running time to $\tilde{O}(\min\{mn^{3/5}, n^{11/5}\})$, compared to the previous $\tilde{O}(mn^{3/5})$ randomized construction. Our algorithm is not only faster in the dense regime but also fully deterministic, conceptually simpler, and easier to implement and analyze. 
\end{abstract}

\section{Introduction}

In graph theory, a fundamental problem is to sparsify a graph while approximately preserving all pairwise distances. This problem arises in numerous practical applications, such as network design, distance oracles, and distributed computing. Additive spanners provide a natural solution to this challenge: they are sparse subgraphs that approximate all pairwise distances within a small additive error.

\begin{definition}[Additive Spanner]
	 A \textbf{$k$-additive spanner} of an unweighted and undirected graph $G = (V, E)$ is a subgraph $H = (V, E')$, where $E' \subseteq E$, such that for all $u, v \in V$, $dist_H(u, v) \le dist_G(u, v) + k$.
\end{definition}
There has been extensive research on the construction and analysis of additive spanners. While the bounds for many other additive spanners are nearly optimal (up to polylogarithmic factors), the optimal edge size and time complexity for constructing $4$-additive spanners remain unresolved. In 2013, Chechik~\cite{10.5555/2627817.2627853} gave a construction of a $4$-additive spanner with $\tilde{O}(n^{7/5})$ edges, but did not analyze the algorithm’s running time. Later, in 2021, Aldhalaan ~\cite{aldhalaan2021fastconstruction4additivespanners} presented a randomized algorithm that constructs a $4$-additive spanner of $\tilde{O}(n^{7/5})$ edges with high probability in $\tilde{O}(mn^{3/5})$ time.

In this paper, we present a new deterministic algorithm for constructing $4$-additive spanners that is faster on dense graphs, while also being conceptually simpler and easier to implement and analyze.

\begin{theorem}[Main Result]\label{mainth}
	There exists a deterministic algorithm that, given a graph $G=(V,E)$ with $n=|V|$ and $m=|E|$, constructs a $4$-additive spanner with $\tilde{O}(n^{7/5})$ edges in $\tilde{O}(\min\{mn^{3/5}, n^{11/5}\})$ time.
\end{theorem}

The remainder of this paper is organized as follows. Section~\textbf{2} reviews related work. Section~\textbf{3} presents our construction for $5$-additive spanners. Section~\textbf{4} provides a detailed analysis of the construction and shows how to obtain a $4$-additive spanner. Finally, Section~\textbf{5} concludes the paper and discusses directions for future work.

\section{Related Work}

For $2$-additive spanners, there is an algorithm that runs in $O(n^2)$ time and returns a $2$-additive spanner with at most $O(n^{3/2})$ edges \cite{knudsen2017additivespannersdistanceoracles}. This size is known to be optimal up to some constant factor \cite{4031374}. 

For $6$-additive spanners there is an algorithm that runs in $\tilde{O}(n^2)$ time and returns a $6$-additive spanner with at most $\tilde{O}(n^{4/3})$ edges \cite{10.1007/978-3-642-14165-2_40}. This $\frac{4}{3}$ exponent is known to be essentially optimal, not only for the 6-additive case but also for a broader range of additive spanners. In particular, it has been shown that for every constant $k$ and any $\epsilon > 0$, there exist graphs that do not admit $k$-additive spanners with $O(n^{4/3 - \epsilon})$ edges~\cite{abboud20174}.

The best-known algorithm for constructing a $4$-additive spanner is a randomized algorithm that produces a subgraph with $\tilde{O}(n^{7/5})$ edges in $\tilde{O}(mn^{3/5})$ running time \cite{aldhalaan2021fastconstruction4additivespanners}. The construction of the $4$-additive spanner proceeds in three stages:

\begin{enumerate}
	\item \textbf{Low-Degree Edges:} The algorithm first adds all edges incident to vertices of degree at most $\tilde{O}(n^{2/5})$.
	
	\item \textbf{BFS Trees from Random Centers:} It then samples a set $S_1$ of $\tilde{O}(n^{2/5})$ vertices uniformly at random. For each $v \in S_1$, the algorithm builds a BFS tree rooted at $v$ and adds all edges in the tree to the spanner.
	
	\item \textbf{Connecting Close Pairs:} We sample another set $S_2$ of size $\tilde{O}(n^{3/5})$. With high probability, every vertex of degree at least $\tilde{O}(n^{2/5})$ has a neighbor in $S_2$. For each such vertex, we add an edge to one of its neighbors in $S_2$ into the spanner. Next, for every pair $u, v \in S_2$, we select a pair $(s,t)$ where $s$ is a neighbor of $u$ and $t$ is a neighbor of $v$, among all such pairs $(s,t)$, we choose the one with minimal distance $\operatorname{dist}(s,t)$, under the condition that some shortest $s$–$t$ path $P$ contains at most $\tilde{O}(n^{1/5})$ edges not yet included in the spanner. Finally, we add the path $P$ to the spanner.
\end{enumerate}

\begin{lemma}[Lemma from \cite{aldhalaan2021fastconstruction4additivespanners}]
	The algorithm above returns a $4$-additive spanner with $\tilde{O}(n^{7/5})$ edges with probability at least $1-\frac{1}{n}$.
\end{lemma}
\begin{remark}
	
	Since all low-degree vertices are fully connected in Step 1, it suffices to ensure that distances between high-degree vertices are preserved within additive $4$.
	
	Intuitively, the construction then ensures coverage in two complementary ways.  
	
	If for some pair of high-degree vertices $u,v$ the shortest $u$–$v$ path contains at most $\tilde{O}(n^{1/5})$ edges not already added in Step~1, then with high probability at least one neighbor of this path will be sampled into $S_1$, so the path will be preserved within additive $2$ in Step~2.  
	
	Otherwise, if such a neighbor is not sampled, Step~3 guarantees that $u$ and $v$ are connected via $S_2$, which ensures that their distance is still preserved within additive $4$.  
\end{remark}

\section{Construction for $5$-Additive Spanners}
Our algorithm is guided by three key insights, which together enable significant improvements.

First, we adopt the degree-based sparsification idea from the $2$-additive spanner construction \cite{knudsen2017additivespannersdistanceoracles}. That algorithm eliminates very high-degree vertices by attaching BFS trees to them, ensuring that all remaining vertices have bounded degree. This preprocessing step drastically simplifies the graph's structure as well as limits the number of edges that subsequent BFS trees can touch. We use the same idea in the first phase of our algorithm, treating it as a deterministic sparsification procedure that helps to improve time complexity.

Second, we reinterpret the randomized $4$-additive spanner construction \cite{aldhalaan2021fastconstruction4additivespanners}. In that algorithm, when the shortest path between two vertices contains many high-degree vertices, preserving distances becomes challenging, and the randomized method handles this probabilistically by sampling multiple random centers. Our approach replaces this randomness with a structured deterministic mechanism: we first select a dominating set of high-degree vertices, then construct dominating sets for all vertices adjacent to shortest paths from each high-degree vertex to any other vertex with sufficient total degree, and finally add covering BFS trees based on these sets. While this may appear computationally intensive, it can be implemented efficiently. Conceptually, this also represents a tradeoff: we forgo the tighter $2$-additive local guarantee in favor of determinism, yielding a seemingly weaker but still robust $4$-additive global bound.

Finally, instead of relying on constrained single-source shortest-path(CSSSP) computations as in prior work, we restrict the algorithm to standard shortest-path computations, making it conceptually simpler. We first reduce the problem to finding a $5$-additive spanner via a bipartite reduction, and make several other subtle refinements. As a result, every step has a clear combinatorial interpretation, and correctness follows directly from simple distance arguments rather than intricate path constraints.

\begin{construction}[$5$-Additive Spanner Construction]\label{cons5}
	Given a graph $G=(V,E)$ with $n=|V|$ and $m=|E|$, initialize $G'=(V',E')=(V,E)$ and $H=(V,\emptyset)$. 
	Define a vertex $v\in V'$ to be \emph{heavy} if $\deg_{G'}(v)\ge n^{2/5}\log ^{3/5}n$, and \emph{light} otherwise.
	The construction proceeds as follows:
	\begin{enumerate}
		\item \textbf{High-degree elimination.}  
		While the vertex with largest degree $v\in V'$ satisfy $\deg_{G'}(v)\ge \frac{n^{3/5}}{\log^{3/5}n}$,  
		\begin{enumerate}
			\item add any BFS tree rooted at $v$ in $G'$ to $H$;
			\item remove $v$ and all its neighbors from $G'$.
		\end{enumerate}
	
		\item \textbf{Light vertices.}  
		For each light vertex $u\in V'$, add all edges incident to $u$ into $H$.
		
		\item \textbf{Dominating set for heavy vertices.}  
		Compute a set $S_1\subseteq V'$ that dominates all heavy vertices in $V'$ and add an edge that connect the vertex to some vertex in $S_1$ for each heavy vertex.  
		
		\item \textbf{BFS Trees for $S_1$.} 
		For each $v\in S_1$, compute a BFS tree $T_v$ rooted at $v$ such that, for every $u\in V'$, the path from $v$ to $u$ minimizes the total vertex degree along the path.  
		Denote by $p_{v,u}$ the parent of $u$ in $T_v$, set
		\[
		f_{v,v}=\deg_{G'}(v), \qquad f_{v,u}=f_{v,p_{v,u}}+\deg_{G'}(u),
		\]
		and let $s_{v,u}$ be the sum of vertex degrees in the subtree of $T_v$ rooted at $u$.
		\item \textbf{Auxiliary bipartite graph $B=(L,R,C)$, initially $L=V',R=C=\emptyset$.}  
		For each $v\in S_1$ and $u\in V'$ satisfying
		\[
		f_{v,u}>n^{3/5}\log ^{2/5}n,\quad f_{v,p_{v,u}}\le n^{3/5}\log ^{2/5}n,\quad s_{v,u}>3n^{3/5}\log ^{2/5}n,
		\]
		include $(v,u)\in R$.  
		For each such pair $(v,u)$, and for every vertex $x\in V'$ that is adjacent to some vertex on the $v$--$u$ path in $T_v$, add $(x,(v,u))\in C$.
		
		\item \textbf{Dominating set for long paths.}  
		Compute a set $S_2\subseteq L$ that dominates all vertices in $R$ in graph $B$, and for each $v\in S_2$, add a BFS tree rooted at $v$ to $H$.
		
		\item \textbf{Adding short paths.}  
		For each $v,u\in S_1$ with $f_{v,u}\le 5n^{3/5}\log ^{2/5}n$, add the path from $v$ to $u$ in $T_v$ to $H$.
	\end{enumerate}
\noindent
\textbf{Dominating Set Computation.}
Both dominating sets $S_1$ and $S_2$ can be obtained using the standard greedy algorithm:
iteratively select the vertex that covers the largest number of currently undominated vertices
until all target vertices are dominated.
\end{construction}

The pseudocode for the construction is presented in the appendix.
\section{Analysis}
\begin{theorem}
	The subgraph $H$ produced by \textbf{Construction}~\autoref{cons5} is a $5$-additive spanner for $G$.
\end{theorem}
\begin{proof}
	Assume, for contradiction, that there exists a pair $(u,v)$ such that
	\[
	\mathrm{dist}_H(u,v) > \mathrm{dist}_G(u,v) + 5,
	\]
	and among all such pairs choose one with minimum $\mathrm{dist}_G(u,v)$.  
	
	Let $P(u,v) = (x_0=u,x_1,\dots,x_\ell=v)$ be a shortest $u$--$v$ path in $G$ minimizing the total vertex degree.
	
	\smallskip
	\noindent\textbf{Case 1:} Some vertex of $P(u,v)$ is deleted in \textbf{Step~1}.  
	Let $x_i$ be the vertex that is first deleted in the process, and let $y$ be the root of the BFS tree added in that iteration($x_i=y$ or $\mathrm{dist}_G(x_i,y)=1$). Then $H$ contains a path from $u$ to $y$ of length at most $i+1$ and a path from $v$ to $y$ of length at most $\ell-i+1$, implying
	\[
	\mathrm{dist}_H(u,v) \le \mathrm{dist}_G(u,v) + 2,
	\]
	contradicting our assumption.  
	Hence no vertex of $P(u,v)$ is removed in \textbf{Step~1}, and thus
	\[
	\mathrm{dist}_{G'}(u,v) = \mathrm{dist}_G(u,v).
	\]
	
	\smallskip
	\noindent\textbf{Case 2:} Either $u$ or $v$ is light vertex. 
	Without loss of generality, say $\deg_{G'}(u)< n^{2/5}\log ^{3/5}n$, then $(u,x_1)\in H$.  
	If $\mathrm{dist}_H(u,v) \ge \mathrm{dist}_G(u,v) + 6$, we also have $\mathrm{dist}_H(x_1,v) \ge \mathrm{dist}_G(x_1,v) + 6$, contradicting the minimality of $\mathrm{dist}_G(u,v)$.  
	
	Thus both $u$ and $v$ must be heavy vertices. Let $s\in S_1$ dominate $u$ and $t\in S_1$ dominate $v$.  
	
	\smallskip
	\noindent\textbf{Case 3:} $f_{s,t} \le 5n^{3/5}\log ^{2/5}n$.
	Then in \textbf{Step~7} the shortest $s$--$t$ path is added to $H$.  
	Since $(s,u),(v,t)\in G'$, we have
	\[
	\mathrm{dist}_{G'}(s,t) \le \mathrm{dist}_{G'}(u,v) + 2,
	\]
	and therefore
	\[
	\mathrm{dist}_H(u,v) \le \mathrm{dist}_{G'}(s,t) + 2 \le \mathrm{dist}_{G'}(u,v) + 4 = \mathrm{dist}_G(u,v) + 4,
	\]
	which is already within $5$.  
	Hence $f_{s,t} > 5n^{3/5}\log ^{2/5}n$.
	
	\smallskip
	\noindent\textbf{Case 4:} $\mathrm{dist}_{G'}(s,t)\le \mathrm{dist}_{G'}(u,v) + 1$. 
	Consider the first vertex $u'$ along the $s$--$t$ path in $T_s$ with $f_{s,u'} > n^{3/5}\log ^{2/5}n$.  
	Since $s_{s,u'} + f_{s,p_{s,u'}} \ge f_{s,t}$, and $f_{s,p_{s,u'}}\le n^{3/5}\log ^{2/5}n$, we have $s_{s,u'} > 4n^{3/5}\log ^{2/5}n$, implying $(s,u')\in R$.  
	Thus some $x\in S_2$ lies adjacent to the $s$--$u'$ path in $T_s$. By adding a BFS tree rooted at each $x\in S_2$ in \textbf{Step~6}, we obtain
	\[
	\mathrm{dist}_H(s,t) \le \mathrm{dist}_{G'}(s,t) + 2.
	\]
	\[
	\mathrm{dist}_H(u,v) \le \mathrm{dist}_H(s,t) + 2\le\mathrm{dist}_{G'}(s,t) + 4 \le \mathrm{dist}_{G'}(u,v) + 5.
	\]
	
	\smallskip
	\noindent\textbf{Case 5:} $\mathrm{dist}_{G'}(s,t) = \mathrm{dist}_{G'}(u,v) + 2$.  
	Then $s\rightarrow u\rightarrow v\rightarrow t$ is a shortest path from $s$ to $t$ in $G'$.  
	Since $f_{s,t} > 5n^{3/5}\log ^{2/5}n$, we have $f_{s,v} + \deg_{G'}(t) > 5n^{3/5}\log ^{2/5}n$, which implies $f_{s,v} > 4n^{3/5}\log ^{2/5}n$.  
	By the similar argument as above, there exists $x\in S_2$ adjacent to the $s$--$v$ path, giving
	\[
	\mathrm{dist}_H(s,v) \le \mathrm{dist}_{G'}(s,v) + 2.
	\]
	Therefore,
	\[
	\mathrm{dist}_H(u,v) \le \mathrm{dist}_H(s,v) + 1 \le \mathrm{dist}_{G'}(u,v) + 4.
	\]
	
	\smallskip
	In all cases,
	\[
	\mathrm{dist}_H(u,v) \le \mathrm{dist}_G(u,v) + 5.
	\]
	Hence $H$ is a $5$-additive spanner for $G$.
\end{proof}
\begin{claim}\label{clm:S1S2}
	We have $|S_1|\le 2n^{3/5}\log^{2/5}n$ and $|S_2|\le 12n^{2/5}\log^{3/5} n$.
\end{claim}
\begin{proof}
	For $S_1$:  
	Suppose $x$ heavy vertices remain undominated.  
	Then there exists a vertex adjacent to at least
	\(\left\lceil\frac{x n^{2/5}\log^{3/5}n}{n}\right\rceil\) of these heavy vertices.  
	Hence after \(\frac{2n^{3/5}}{\log^{3/5}n}\) rounds the number of undominated heavy vertices decreases to at most \(x/2\).  
	Repeating this halving process \(\log n\) times leaves none undominated, yielding
	\(|S_1|\le 2n^{3/5}\log^{2/5} n\).
	
	For $S_2$:  
	Each vertex \((u,v)\in R\) satisfies \(f_{u,v}>n^{3/5}\log^{2/5}n\).  
	Because each vertex can be adjacent to at most three vertices in a shortest path,
	\((u,v)\) must be adjacent to at least \(n^{3/5}\log^{2/5}n/3\) vertices in \(L\).  
	If \(x\) vertices in \(R\) remain undominated, there are at least 
	\((xn^{3/5}\log^{2/5}n)/3\) edges between them and \(L\).  
	Thus some vertex in \(L\) is adjacent to at least \(x\log^{2/5}n/(3n^{2/5})\) undominated vertices.  
	Hence after \(\frac{6n^{2/5}}{\log^{2/5}n}\) rounds 
	the number of undominated vertices drops to at most \(x/2\).  
	Starting from \(x\le n^2\), after \(12n^{2/5}\log^{3/5} n\) rounds all the vertices in \(R\) are dominated, giving
	\(|S_2|\le 12n^{2/5}\log^{3/5} n\).
\end{proof}
\begin{theorem}\label{edgenum}
	The spanner $H$ contains at most $O(n^{7/5}\log^{3/5} n)$ edges.
\end{theorem}
\begin{proof}
	\textbf{Step~1:}  
	Each BFS tree contributes at most $n$ edges.  
	In each iteration we delete at least $\frac{n^{3/5}}{\log^{3/5}n}$ vertices, and each vertex is deleted at most once.  
	Hence Step~1 adds at most
	\[
	O\left(\frac{n}{\frac{n^{3/5}}{\log^{3/5}n}}\right) \cdot n = O(n^{7/5}\log^{3/5}n)
	\]
	edges.
	
	\smallskip
	\textbf{Step~2:}  
	At most $O(n^{7/5}\log^{3/5}n)$ edges are added.
	
	\smallskip
	\textbf{Step~3:}  
	This step contributes at most $n$ edges.
	
	\smallskip
	\textbf{Step~6:}  
	We add BFS trees for at most $|S_2|\le 12n^{2/5}\log^{3/5}n$ vertices,  
	each tree contains at most $n$ edges, giving  
	\[
	O(n\cdot n^{2/5}\log^{3/5}n) = O(n^{7/5}\log^{3/5}n).
	\]
	
	\smallskip
	\textbf{Step~7:}  
	There are at most $O(n^{6/5}\log^{4/5}n)$ such pairs.  
	For each pair we add a path whose total degree is at most $5n^{3/5}\log^{2/5}n$.  
	Note that edges incident to vertices of degree less than $n^{2/5}\log^{3/5}n$ have already been included earlier,  
	so each such path contributes at most $\frac{5n^{1/5}}{\log^{1/5}n}$ new edges.  
	Thus Step~7 adds at most
	\[
	O\left(n^{6/5}\log^{4/5}n \cdot \frac{n^{1/5}}{\log^{1/5}n}\right)
	= O(n^{7/5}\log^{3/5}n)
	\]
	edges.
	
	\smallskip
	Combining all steps, $H$ contains $O(n^{7/5}\log^{3/5}n)$ edges.
\end{proof}

\begin{theorem}\label{time}
	The construction of $H$ can be implemented in 
	\[
	O\!\left(\min\left(mn^{3/5}\log^{2/5}n,\;\frac{n^{11/5}}{\log^{1/5}n}\right)\right)
	\]
	time.
\end{theorem}
\begin{proof}
	We may assume $m>n^{7/5}$; otherwise we can simply add all the edges to $H$ directly.
	
	\smallskip
	\textbf{Step~1:}  
	Each BFS takes $O(m)$ time and removes at least $\Theta(m/n)$ vertices,  
	yielding a total cost of $O(n^2)$.
	
	\smallskip
	\textbf{Step~2:}  
	This step requires $O(m)$ time.
	
	\smallskip
	\textbf{Step~3:}  
	This step can be completed in $O(n^2)$ time.
	
	\smallskip
	Let $D=|E'|$. Clearly $D\le m$.  
	Since each vertex has degree at most $O\big(\frac{n^{3/5}}{\log^{3/5}n}\big)$ in $G'$,  
	\[
	D\le \frac{n^{8/5}}{\log^{3/5}n}.
	\]
	\smallskip
	
	\textbf{Step~4:}  
	For each $v\in S_1$ (with $|S_1|=O(n^{3/5}\log^{2/5}n)$), we run a BFS, costing
	\[
	O(|S_1|\cdot D).
	\]
	
	\smallskip
	\textbf{Step~5:}  
	For each $u\in S_1$, adding $(u,v)$ to $R$ requires that the subtree rooted at $v$ has 
	size at least $3n^{3/5}\log^{2/5}n$.  
	Moreover, the subtrees corresponding to different $(u,v)$ are disjoint.  
	Hence
	\[
	|R| \le |S_1| \cdot \frac{D}{3n^{3/5}\log^{2/5}n},
	\]
	and the bipartite graph has at most
	\[
	|C| \le |S_1|\cdot D
	\]
	edges.
	
	\smallskip
	\textbf{Step~6:}  
	Computing $S_2$ takes $O(|C|)$ time and 
	constructing BFS trees for all vertices in $S_2$ costs $O(|S_2|\cdot D)$. Thus $O(|S_1|D)$ in this step.
	
	\smallskip
	\textbf{Step~7:}  
	The BFS trees are already computed, so extracting the paths requires at most $O(n^2)$ time.
	
	\smallskip
	Therefore, combining all the steps, the total running time is
	\[
	O\!\left(\min\left(mn^{3/5}\log^{2/5}n,\;\frac{n^{11/5}}{\log^{1/5}n}\right)\right).
	\]
\end{proof}

\begin{theorem}\label{reduc}
	Suppose there exists an algorithm that, for any graph with $n$ vertices and $m$ edges, 
	computes a $(2k+1)$-additive spanner with $O(f(n,m))$ edges in $O(g(n,m))$ time.  
	Then there exists an algorithm that computes a $(2k)$-additive spanner with 
	$O(f(2n,2m))$ edges in $O(g(2n,2m)+m)$ time.
\end{theorem}

\begin{proof}
	Given a graph $G=(V,E)$, construct a bipartite graph
	\[
	G_0=(L,R,E_0),\qquad L=R=V,
	\]
	where for each $(u,v)\in E$, we add $(L_u,R_v)$ and $(R_u,L_v)$ to $E_0$.
	
	\smallskip
	Run the assumed algorithm on $G_0$ to obtain a $(2k+1)$-additive spanner $H_0$.
	We now construct $H$ for $G$ by including $(u,v)\in E$ whenever $(L_u,R_v)\in H_0$.
	
	\smallskip
	Consider any pair $u,v\in V$:
	
	\begin{itemize}
		\item \textbf{Case 1:} The shortest $u$--$v$ path in $G$ has even length:
		\[
		(x_0=u, x_1,\dots,x_{2\ell}=v).
		\]
		Then in $G_0$ there is a corresponding path
		\[
		(L_{x_0}, R_{x_1}, L_{x_2},\dots,L_{x_{2\ell}}).
		\]
		Since $H_0$ is a $(2k+1)$-additive spanner, there exists a path $P$ in $H_0$ from $L_{x_0}$ to $L_{x_{2\ell}}$
		of length at most $2\ell+2k+1$.  
		But $H_0$ is bipartite, so any path between two vertices in $L$ must have even length; hence
		\[
		\mathrm{dist}_{H_0}(L_{x_0},L_{x_{2\ell}})\le 2\ell+2k.
		\]
		By construction of $H$, each $(L_u,R_v)\in H_0$ corresponds to $(u,v)\in H$,
		thus
		\[
		\mathrm{dist}_{H}(u,v)\le 2\ell+2k=\mathrm{dist}_{G}(u,v)+2k.
		\]
		
		\item \textbf{Case 2:} The shortest $u$--$v$ path in $G$ has odd length:
		\[
		(x_0=u, x_1,\dots,x_{2\ell+1}=v),
		\]
		yielding the $G_0$ path
		\[
		(L_{x_0}, R_{x_1},\dots,R_{x_{2\ell+1}}).
		\]
		There exists a path in $H_0$ from $L_{x_0}$ to $R_{x_{2\ell+1}}$ 
		of length at most $2\ell+2k+2$.  
		Again, $H_0$ is bipartite, so any path connecting $L$ to $R$ has odd length, giving
		\[
		\mathrm{dist}_{H_0}(L_{x_0},R_{x_{2\ell+1}})\le 2\ell+2k+1.
		\]
		Therefore,
		\[
		\mathrm{dist}_{H}(u,v)\le 2\ell+2k+1 = \mathrm{dist}_{G}(u,v)+2k.
		\]
	\end{itemize}
	
	Hence, for every pair $(u,v)$,
	\[
	\mathrm{dist}_H(u,v) \le \mathrm{dist}_G(u,v) + 2k,
	\]
	so $H$ is a $(2k)$-additive spanner of $G$.
	The edge bound is $O(f(2n,2m))$, and the running time is $O(g(2n,2m)+m)$.
\end{proof}

\begin{theorem}[Restatement of Theorem~\ref{mainth}]
	There exists a deterministic algorithm that, given a graph $G=(V,E)$ with $n=|V|$ and $m=|E|$, constructs a $4$-additive spanner with $\tilde{O}(n^{7/5})$ edges in $\tilde{O}(\min\{mn^{3/5}, n^{11/5}\})$ time.
\end{theorem}

\begin{proof}
	We apply Theorem~\ref{reduc} with $k=2$, combining it with our $5$-additive spanner construction. 
	By Theorems~\ref{edgenum} and~\ref{time}, this yields a $4$-additive spanner containing $\tilde{O}(n^{7/5})$ edges and computable in $\tilde{O}(\min\{mn^{3/5}, n^{11/5}\})$ time, as claimed.
\end{proof}
\section{Conclusion}
We presented a deterministic algorithm for constructing a $4$-additive spanner 
with $\tilde{O}(n^{7/5})$ edges, computable in 
$\tilde{O}(\min\{mn^{3/5},\,n^{11/5}\})$ time. 
Our approach extends the high-degree vertex elimination technique previously 
used for $2$-additive spanners and replaces random sampling with 
dominating-set based derandomization, yielding a substantial improvement over existing results. Through the analysis of 5-additive spanners and the subsequent reduction, we bypass the need to solve the constrained single-source shortest path problem. Instead, the algorithm only requires BFS tree computations, which makes the algorithm simpler.

Several questions remain open. 
An important direction is to narrow the gap between the current 
$\tilde{O}(n^{7/5})$ upper bound and the known 
$\Omega(n^{4/3})$ lower bound on the size of $4$-additive spanners. 

Another challenge is to reduce the construction time to $O(n^{2})$, 
matching the fastest known algorithms for other spanner families.
\section*{Acknowledgments}

I am grateful to Jason Li and Yusi Chen for valuable discussions on $4$-additive spanners.

\begingroup
\setlength{\itemsep}{0.5em} 
\nocite{*}                 
\bibliographystyle{plain}     
\bibliography{ref}
\endgroup

\appendix
\section*{Appendix A: Pseudocode}

\begin{algorithm}[H]
	\caption{Deterministic Construction of 5-Additive Spanner}\label{pse5}
	\begin{algorithmic}[1]
		\State \textbf{Input:} Unweighted undirected graph $G = (V, E)$
		\State \textbf{Output:} 5-additive spanner $H = (V, E_H)$
		\State Initialize $G' \gets G$, $H \gets (V, \emptyset)$
		\vspace{0.5em}
		\State \textbf{// Step 1: High-degree elimination}
		\While{there exists vertex $v \in V(G')$ with $\deg_{G'}(v) \ge n^{3/5}/\log^{3/5} n$}
		\State Add a BFS tree rooted at $v$ in $G'$ to $H$
		\State Remove $v$ and all its neighbors from $G'$
		\EndWhile
		\vspace{0.5em}
		\State \textbf{// Step 2: Light vertices}
		\For{each vertex $u$ in $G'$ with $\deg_{G'}(u) < n^{2/5}\log^{3/5} n$}
		\State Add all edges incident to $u$ to $H$
		\EndFor
		\vspace{0.5em}
		\State \textbf{// Step 3: Dominating set for heavy vertices}
		\State Let $S_1 \subseteq V(G')$ be a dominating set of heavy vertices in $G'$
		\For{each heavy vertex $v\in G'$}
		\State Add edge $(v, u)$ to $H$ for some $u \in S_1$ such that $(v,u) \in G'$
		\EndFor
		\vspace{0.5em}
		\State \textbf{// Step 4: Construct specialized BFS trees from $S_1$}
		\For{each $v \in S_1$}
		\State Construct BFS tree $T_v$ rooted at $v$ minimizing total vertex degree on paths
		\State For each $u$, compute $f_{v,u}$ = total degree from $v$ to $u$ in $T_v$
		\State For each $u$, compute $s_{v,u}$ = total degree in subtree rooted at $u$ in $T_v$
		\EndFor
		\vspace{0.5em}
		\State \textbf{// Step 5: Build auxiliary bipartite graph}
		\State Initialize bipartite graph $B = (L, R, C)$ with $L = V(G')$, $R = \emptyset$, $C = \emptyset$
		\For{each $v \in S_1$, $u \in V(G')$}
		\If{$f_{v,u} > n^{3/5}\log^{2/5} n$, $f_{v,p_{v,u}} \le n^{3/5}\log^{2/5} n$, and $s_{v,u} > 3n^{3/5}\log^{2/5} n$}
		\State Add $(v, u)$ to $R$
		\For{each $x$ adjacent to some vertex on the $v$-$u$ path in $T_v$}
		\State Add edge $(x, (v, u))$ to $C$
		\EndFor
		\EndIf
		\EndFor
		\vspace{0.5em}
		\State \textbf{// Step 6: Add BFS trees from dominators of long paths}
		\State Let $S_2 \subseteq L$ be a dominating set of $R$ in $B$
		\For{each $x \in S_2$}
		\State Add BFS tree rooted at $x$ in $G'$ to $H$
		\EndFor
		\vspace{0.5em}
		\State \textbf{// Step 7: Add short paths}
		\For{each pair $v,u\in S_1$ such that $f_{v,u} \le 5n^{3/5}\log^{2/5} n$}
		\State Add the path from $v$ to $u$ in $T_v$ to $H$
		\EndFor
		\vspace{0.5em}
		\State \Return $H$
	\end{algorithmic}
\end{algorithm}

\begin{algorithm}[H]
	\caption{4-Additive Spanner Construction}
	\begin{algorithmic}[1]
		\State \textbf{Input:} Graph $G = (V, E)$
		\State \textbf{Output:} A $4$-additive spanner $H$ of $G$
		\vspace{0.5em}
		\State Let $L \gets V$, $R \gets V$
		\State Construct bipartite graph $G_0 = (L, R, E_0)$ as follows:
		\For{each edge $(u, v) \in E$}
		\State Add edges $(L_u, R_v)$ and $(R_u, L_v)$ to $E_0$
		\EndFor
		\vspace{0.5em}
		\State Run \textbf{Algorithm} \autoref{pse5} on $G_0$ to obtain $5$-additive spanner $H_0$
		\State Initialize $H \gets \emptyset$
		\For{each edge $(u, v) \in E$}
		\If{$(L_u, R_v) \in E(H_0)$ or $(R_u, L_v) \in E(H_0)$}
		\State Add edge $(u, v)$ to $H$
		\EndIf
		\EndFor
		\vspace{0.5em}
		\State \Return $H$
	\end{algorithmic}
\end{algorithm}

\end{document}